\documentclass[conference]{IEEEtran}
\IEEEoverridecommandlockouts
\usepackage{cite}
\usepackage{amsmath,amssymb,amsfonts}
\usepackage{graphicx}
\usepackage{textcomp}
\usepackage{xcolor}
\usepackage{amsmath}
\usepackage{float}
\usepackage{multicol}
\usepackage{amssymb}
\usepackage{graphicx}
\usepackage{listings}
\usepackage{cite}
\usepackage{amssymb}
\usepackage{framed}
\newcommand{\vect}[1]{\vec{\mathbf{#1}}}
\newcommand{\ZZ}{\mathbb{Z}}
\newcommand{\fname}{\mathsf}
\newcommand{\aname}[1]{\mathbf{\fname{#1}}}
\newcommand{\vname}{\mathsf}
\newcommand{\samplefrom}{\xleftarrow{\$}}

\newcommand{\concat}{\parallel}

\newcommand{\Com}{\mathsf{Com}}

\newcommand{\val}{\vname{\$val}}
\newcommand{\sval}[1]{\vname{\$ #1}}
\newcommand{\party}{\mathcal{P}}

\newcommand{\valset}{\mathbb{V}}
\newcommand{\inp}{\vname{in}}
\newcommand{\coin}{\vname{coin}}
\newcommand{\manager}{\mathcal{M}}

\renewcommand{\Pr}[1]{\fname{Pr}\Big[#1\Big]}

\newtheorem{theorem}{Theorem}
\newtheorem{definition}[theorem]{Definition}
\newcommand{\qedsymbol}{\ensuremath{\blacksquare}}
\newenvironment{proof}{\noindent{\bf Proof:~~}}{\qedsymbol}

\graphicspath{ {figures/} }
\def\BibTeX{{\rm B\kern-.05em{\sc i\kern-.025em b}\kern-.08em
    T\kern-.1667em\lower.7ex\hbox{E}\kern-.125emX}}
\begin{document}

\title{zkHawk: Practical Private Smart Contracts from MPC-based Hawk
\thanks{This publication has emanated from research conducted with the financial support of Science Foundation Ireland grants 13/RC/2106 (ADAPT) and 17/SP/5447 (FinTech Fusion). This work was also supported in part by Science Foundation Ireland grant 13/RC/2094 (Lero).}
}

\author{\IEEEauthorblockN{Aritra Banerjee}
\IEEEauthorblockA{
ADAPT Centre\\
\textit{School of Comp Sci \& Stats}\\
\textit{Trinity College Dublin}\\
Dublin, Ireland \\
abanerje@tcd.ie
}
\and
\IEEEauthorblockN{Michael Clear}
\IEEEauthorblockA{
\textit{School of Comp Sci \& Stats}\\
\textit{Trinity College Dublin}\\
Dublin, Ireland \\
clearm@tcd.ie
}
\and
\IEEEauthorblockN{Hitesh Tewari}
\IEEEauthorblockA{
\textit{School of Comp Sci \& Stats}\\
\textit{Trinity College Dublin}\\
Dublin, Ireland \\
htewari@tcd.ie}
}

\maketitle

\begin{abstract}
Cryptocurrencies have received a lot of research attention in recent years following the release of the first cryptocurrency Bitcoin. With the rise in cryptocurrency transactions, the need for smart contracts has also increased. Smart contracts, in a nutshell, are digitally executed contracts wherein some parties execute a common goal. The main problem with most of the current smart contracts is that there is no privacy for a party's input to the contract from either the blockchain or the other parties.  Our research builds on the Hawk project that provides transaction privacy along with support for smart contracts. However, Hawk relies on a special trusted party known as a manager, which must be trusted not to leak each party's input to the smart contract. In this paper, we present a practical private smart contract protocol that replaces the manager with an MPC protocol such that the function to be executed by the MPC protocol is relatively lightweight, involving little overhead added to the smart contract function, and uses practical sigma protocols and homomorphic commitments to prove to the blockchain that the sum of the incoming balances to the smart contract matches the sum of the outgoing balances.
\end{abstract}


\section{Introduction}
Cryptocurrencies  are  generally  decentralized and  based  on  a  public distributed ledger called a blockchain. Many cryptocurrencies that followed Bitcoin, including Ethereum, do not offer transaction privacy. In a nutshell, transaction privacy hides the origin of a transaction (the sending party), along with the amounts transacted, such that a third party inspecting the blockchain cannot track the flow of money between accounts. Zcash \cite{sasson2014zerocash} and Monero \cite{miller2017empirical} are two popular cryptocurrencies that provide transaction privacy. A major limitation of these cryptocurrencies is that they do not support smart contracts. In short, a smart contract is a program that determines a set of deposits and withdrawals to/from a set of (anonymous) accounts. A smart contract is created by a set of parties to  facilitate  a  common  goal. We call these set of parties the \textit{participant}s of the smart contract.
The  absence  of  support  for  smart  contracts  in  existing cryptocurrencies with  transaction  privacy  inspired  the  Hawk  project \cite{kosba2016hawk}. Hawk  obtains transaction privacy with support for smart contracts, albeit with the trust assumption of a special type of entity known as a \textit{manager}. The manager sees the inputs of each participant in the smart contract and must be explicitly trusted not to leak them. There have been other recent developments in privacy where the smart contract execution is done off-chain like in Arbitrum\cite{kalodner2018arbitrum} which used the concept of Virtual Machines (VMs) along with the trust assumption of a \textit{manager} as well. But Arbitrum provided only partial privacy, i.e, the input of the parties are hidden from blockchain but not from each other. The goal of this research is to build upon Hawk so that along with support for transaction privacy and smart contracts, we additionally obtain  privacy  for  the  participant’s inputs (i.e. they are hidden from the blockchain and each other), without the trust assumption of the \textit{manager}.\\
Our starting point is an observation made by the authors of the Hawk protocol \cite{kosba2016hawk}, namely that the manager can be replaced by running a multi-party computation (MPC) protocol between the parties. However the authors of the Hawk paper point out that this approach would be presently impractical. Indeed it can be readily seen that by applying MPC to the design of Hawk as it is incurs the prohibitively expensive overhead of executing a zk-SNARK proof \cite{ben2014succinct,groth2010short,banerjee2020demystifying} within an MPC program  (a circuit in practice). We therefore propose to remove the zk-SNARK proof from the MPC program. However this leaves us with a considerable challenge - How do we prove to the blockchain that the sum of the in-going account balances in a smart contract is equal to the sum of the outgoing balances? As a solution, we will compute a relatively practical proof of knowledge proof within the MPC program such that all parties contribute to the witness. The resulting effect is to guarantee that the difference between the sum of the input balances and the sum of the outgoing balances is zero.
An MPC protocol \cite{canetti1996adaptively} enables a collection of parties to interact with each other in several ``rounds" of communication in order to compute a function $f$ and learn the output $y = f(x_1,x_2,\hdots,x_n)$ where $x_i$ is party $i's$ input. Such that even if up to $t$ parties are malicious (for some collusion tolerance $t$); they cannot learn the other parties' inputs i.e. they are kept secret.
We distinguish between two types of input/output privacy, which are:
\begin{itemize}
    \item \textbf{Weak Input/Output Privacy} is defined such that the execution of a smart contract does not reveal any parties' inputs/outputs to the public. But the inputs/outputs of each party are not hidden from each other.
    \item \textbf{Strong Input/Output Privacy} is defined such that parties’ inputs/outputs are not leaked to both the public and to each of the other parties.
\end{itemize}
Furthermore, we distinguish between two types of PSC evaluation: non-interactive PSC (NI-PSC) and interactive PSC (I-PSC). In the former, the blockchain executes the smart contract function and no interaction is needed between the contract participants. In contrast, interactive PSC involves executing the smart contract function off-chain by all parties interacting for the purpose of executing an MPC protocol to compute the smart contract function. This is the setting primarily explored in this paper. We do however briefly touch upon non-interactive PSC in Section~\ref{sec:nipsc} but we leave further development to future work.

We consider a cryptocurrency that supports transaction privacy. More precisely, we expect, for example, to inherit a blockchain that implements the $\vname{Blockchain}_{\vname{cash}}$ program in  \cite{kosba2016hawk} and a user program that implements the $\vname{UserP}_\vname{cash}$ protocol in \cite{kosba2016hawk} i.e. the operations of mint and pour from ZCash/Hawk that provides transaction privacy. To avoid implementing these operations in this work, we interface with such a construction through a commitment scheme that is defined by the PSC evaluation protocol.
First, we formalize an (interactive) PSC evaluation protocol in this paper as shown in Section \ref{sec:fd}. In the formal definitions sections we also specify the set of input coins $\coin_1, \hdots, \coin_n$ which are commitments to hidden values $\val_1, \hdots, \val_n$ respectively.
Next, we construct a concrete PSC evaluation protocol as shown in Section \ref{sec:PSCEP} that allows the contract participants to evaluate a smart contract function off-chain that yields a new set of coins $\coin'_1, \hdots, \coin'_n$ that hide the output values $\val'_1, \hdots, \val'_n$ from the smart contract such that the blockchain can be sure that $\sum_{i \in [n]} \val'_i = \sum_{i \in [n]} \val_i$. Finally, we prove our protocol $t$-secure for collusion tolerance $t < n$ assuming the hardness of the discrete logarithm problem in the random oracle model (Section \ref{sec:security}). Although this security notion is weaker than security in the universal composability framework, it shows our protocol resists a wide variety of attacks from a malicious adversary that can \emph{statically} corrupt $t < n$ parties. We intend to prove full security in the UC framework \cite{canetti2001universally} in future work.

One of the primary goals of our protocol is to minimize the computation overhead of the MPC program to ensure practicality. With this in mind, we use simple and lightweight cryptographic tools such as standard sigma protocols that can be computed efficiently. We also exploit a simple trick that involves decomposing the output coin commitment into commitments to its individual bits, then for each bit, sending both candidate commitments (a commitment to a zero and a commitment to a one) in random order to the blockchain, and choosing between the candidates using a simple multiplexer circuit in the MPC program based on the output value of the coin. Therefore, this idea simultaneously functions as a range proof and a way to reduce computation in the MPC program (we can use SPDZ\cite{araki2018generalizing}) to a simple multiplexer circuit. The associated cost of this approach is the requirement of two NIZK proofs for every bit. These proofs are however not computed within the MPC program so the approach remains practical.
\subsection{Example of a Private Smart Contract}
Consider a sealed bid auction as a Private Smart Contract as shown in Figure \ref{fig:auction}. In a sealed-bid auction, one of the participants is the seller and the other participants are bidders. Each bidder submits a bid and the bidder with the highest bid wins the auction. All bids are kept private.\\
Consider a function $f$ that implements this smart contract. Suppose there are $k$ bidders, and in total $n = k + 1$ parties (recall the seller is one of the parties). The seller is designated as the first party. We define:
\begin{itemize}
    \item $f((\sval{seller}, \cdot), (\sval{bidder}_1, \cdot), \hdots, (\sval{bidder}_k, \cdot))$
    \begin{itemize}
        \item $\sval{highest} \gets 0$
        \item
        $\vname{winner} \gets 0$
        \item
        For $i \in \{1, \hdots, k\}:$
        \begin{itemize}
            \item If $\sval{bidder}_i > \sval{highest}$
            \begin{itemize}
                \item $\sval{highest} \gets \sval{bidder}_i$
                \item
                $\vname{winner} \gets i$
            \end{itemize}
        \end{itemize}
        \item
        $\sval{seller}' \gets \sval{seller} + \sval{highest}$
        \item
        $\sval{bidder}'_{\vname{winner}} \gets 0$
        \item
        For $i \in \{1, \hdots, k\} \setminus \{\vname{winner}\}:$
        \begin{itemize}
            \item 
            $\sval{bidder}'_i \gets \sval{bidder}_i$
        \end{itemize}
        \item
        Return $(\sval{seller}', \sval{bidder}'_1, \hdots, \sval{bidder}'_k, \vname{out} \\ := \vname{winner})$
    \end{itemize}
\end{itemize}

\begin{figure}[H]
\centering
\includegraphics[width=3.5in]{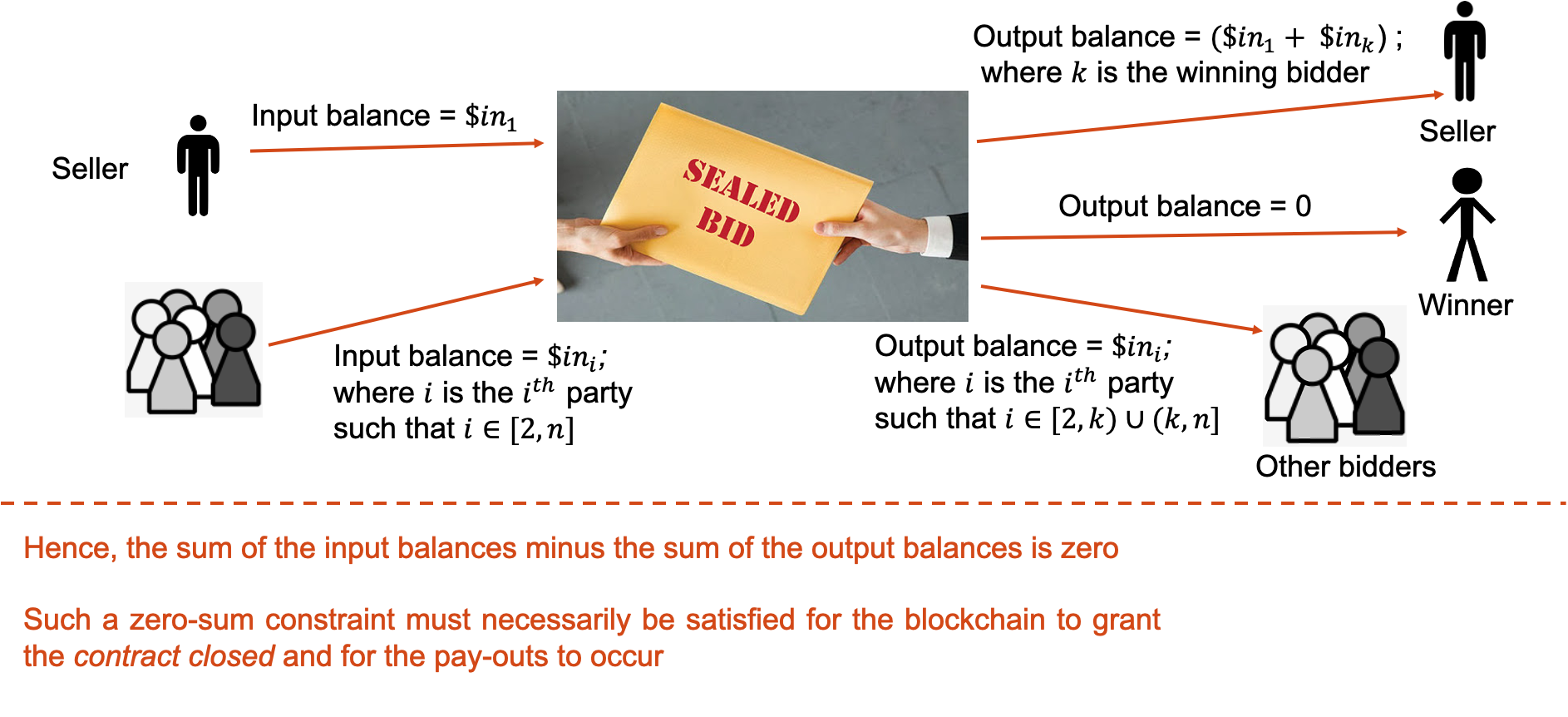}
\caption{Example of a Private Smart Contract: A Sealed Bid Auction}
\label{fig:auction}
\end{figure}

\subsection{Overview of Interactive Private Smart Contracts}
Private Smart Contract (PSC) execution involves three phases as shown in Figure \ref{fig:smart}. The first, which we call the \textit{freeze} phase (a term borrowed from Hawk), is where each participant sends its input coin to the blockchain. The blockchain ``freezes" the coins so that they cannot be spent and records them for use later. The second phase is \textit{computation}. This occurs off-chain and involves the parties running an MPC protocol between them. The third phase is what we call \textit{finalization} where at least one party must notify the blockchain of the result of the MPC program. There are a set of output coins resulting from smart contract execution (one for each party). The blockchain has to check that the sum of the input coins is equal to the sum of the output coins. Once this is asserted, the blockchain can award the output coin to each party; we call this \textit{contract closure}\footnote{A contract is said to be closed when the zero-sum constraint is evaluated to be satisfied and the payouts are made.}.

\begin{figure}[H]
\centering
\includegraphics[width=3.5in]{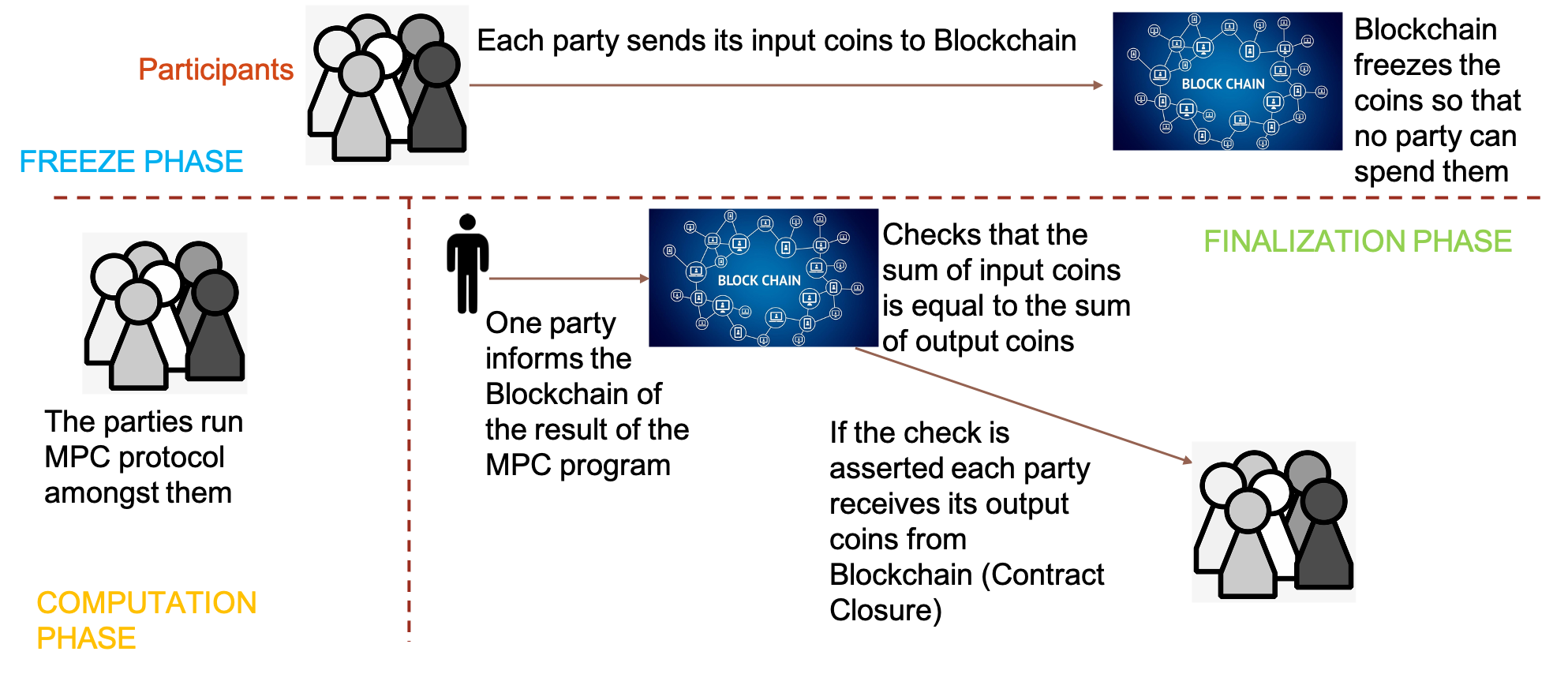}
\caption{How Interactive Private Smart Contracts will Work}
\label{fig:smart}
\end{figure}

\section{Formal Definitions}\label{sec:fd}
We distinguish amounts of currency with a leading ``\$" (dollar sign) symbol before the associated variable name or literal value to enhance readability. We define the set of valid currency values $\valset$ as the set of non-negative integers less than a strict upper bound $L$. More formally, we have that $\valset := \{\val \in \ZZ : 0 \leq \val < L\}$. For convenience, we choose $L$ as a power of 2 i.e. we let $L = 2^\ell$ for some positive integer $\ell$.

\begin{definition}\label{def:scf}
An $n$-party smart contract is an $n$-ary function $f : (\valset \times \{0, 1\}^\ast)^n \to (\valset^n \times \{0, 1\}^\ast) \cup \{\bot\}$ satisfying the following property:
\begin{itemize}
    \item For all choices of $\val_1, \hdots, \val_n \in \mathbb{V}$ and $\inp_1, \hdots, \inp_n \in \{0, 1\}^\ast$, one of the following statements is true
    \begin{enumerate}
        \item $t = \bot$ (failure)
        \item $t = (\val'_1, \hdots, \val'_n, \vname{out}) \land $
        \[\sum_{i \in [n]} \val'_i - \sum_{i \in [n]} \val_i = 0 \; \text{(zero-sum constraint)}\]
    \end{enumerate}
where $t = f((\val_1, \inp_1), \hdots, (\val_n, \inp_n))$.
\end{itemize}
\end{definition}

\subsection{Coins}
Following on from Zcash \cite{sasson2014zerocash} and Hawk \cite{kosba2016hawk}, we define a \emph{coin} as a commitment to some value $\val$ with randomness $r$ using some binding and hiding commitment scheme $\Com$. We write this as $\coin = \Com(\val; r)$. We assume the inclusion of a mechanism to provide transaction privacy in a manner such as Zcash and Hawk, but we endeavor to sidestep repeating the details here by defining an appropriate interface between such a mechanism for transaction privacy and a PSC evaluation protocol. We achieve this by assuming that the value associated with the input coin $\coin_i$ of each party $\party_i$ is private; that is, hidden from the blockchain and the other parties. Execution of the smart contract produces an output coin $\coin'_i$ for each party $\party_i$ whose value is also hidden from the blockchain and the other parties. A particular PSC evaluation protocol, defined momentarily, specifies the commitment scheme that it uses to bind a value $\val$ to a $\coin$ using randomness $r$.

\subsection{PSC Evaluation Protocol}
We now present the notion of a private smart contract (PSC) evaluation protocol. This is a protocol that consists of a blockchain program, which handles three types of messages $(\vname{freeze}, \cdot)$, $(\vname{compute}, \cdot)$ and $(\vname{finalize}, \cdot)$ and a collection of user parties $\party_1, \hdots, \party_n$. We defer to the full version a formal generalized definition of a PSC evaluation protocol using the blockchain model of cryptography \cite{kosba2016hawk} in the universal composability (UC) \cite{canetti2001universally} framework. Our definition here is more specialized as it specifically addresses the case of interactive PSC, which is what our protocol targets in this paper. In interactive PSC, computation of the smart contract function occurs off-chain, most likely using a multi-party computation protocol. Several elements of our definition are inspired by the blockchain model \cite{kosba2016hawk} but our model here is effectively ``stripped-down" and tailored to the task in hand. We require a special trusted incorruptible entity $\manager$ that is defined as follows. If $\manager$ receives a message $(\vname{input}, F, x_i)$ from every party $\party_i$ for $i \in [n]$, then it privately computes $y = F(x_1, \hdots, x_n)$ and sends $y$ to every party $\party_i$ with the message $(\vname{output}, F, y)$. Note that there is an authenticated and private channel between $\manager$ and every party $\party_i$. Such an entity $\manager$ allows us to model a correct and secure idealized MPC protocol.

\begin{definition}\label{def:pscprot}
A PSC evaluation protocol $\pi$ (in the interactive setting) is a tuple $(\Com, B, U)$  where $\Com$ is a commitment scheme used to establish input and output coins, $B = (\aname{Init}, \aname{Freeze}, \aname{Finalize})$ is the blockchain program (modelled as a tuple of stateful PPT algorithms) and $U = (\aname{Init}, \aname{Create}, \aname{Freeze}, \aname{PrepareCompute}, \aname{Finalize})$ is the user program. The program $B$ is passed to the blockchain program wrapper $B'$, defined in Figure~\ref{fig:blockchain_wrapper}, to produce a blockchain functionality. The program $U$ is passed to the user program wrapper $U'$, defined in Figure~\ref{fig:user_wrapper} to specify user behavior during execution of the protocol $\pi$.
\end{definition}

\begin{figure}[!ht]
%\begin{framed}
\begin{small}
\begin{center}
\hspace{-19pt}
Blockchain program wrapper $B'(B):$
\begin{tabular}{l}
\begin{minipage}{3in}
\begin{tabbing}
123\=12\=12\=12\=12\=\kill
\textbf{Init:} \\
\> $\vname{Contracts} \gets \{\}$\\
\> Call $\aname{B.Init()}$ (i.e. Call program $B$)\\
\> Send blockchain state to $\mathcal{A}$.
\end{tabbing}
\end{minipage} \vspace{6pt} \\
\begin{minipage}{3in}\vspace{6pt}
\begin{tabbing}
123\=12\=12\=12\=12\=\kill
\textbf{Freeze:} Upon receiving $\vname{msg} := (\vname{freeze}, \vname{id}, P, \cdot)$ from $\party$:\\
\> Send $(\vname{msg}, \party)$ to $\mathcal{A}$\\
\> Assert $\party \in P$\\
\> If $(\vname{id}, \cdot, \cdot, \cdot) \notin \vname{Contracts}$:\\
\>\> $\vname{Contracts} \gets \vname{Contracts} \cup \{(\vname{id}, P, P' := \{\}, \vname{freeze})\}$ \\
\> Assert $(\vname{id}, P, P', t := \vname{freeze}) \in \vname{Contracts}$\\
\> Assert $\party \notin P'$\\
\> If $\aname{B.Freeze}(\vname{msg}, \party) = 1$: (i.e. Call program $B$)\\
\>\> If $P' \cup \{\party\} = P$: \\
\>\>\> $t \gets \vname{compute}$\\
\>\> Replace $(\vname{id}, P, P', \vname{freeze})$ in $\vname{Contracts}$\\
\>\>\> with $(\vname{id}, P, P' \cup \{\party\}, t)$\\
\> Send blockchain state to $\mathcal{A}$
\end{tabbing}
\end{minipage} \vspace{6pt} \\
\begin{minipage}{3in}\vspace{6pt}
\begin{tabbing}
123\=12\=12\=12\=12\=\kill
\textbf{Compute:} Computation is performed off-chain
\end{tabbing}
\end{minipage} \vspace{6pt} \\
\begin{minipage}{3in}\vspace{6pt}
\begin{tabbing}
123\=12\=12\=12\=12\=\kill
\textbf{Finalize:} Upon receiving $\vname{msg} := (\vname{finalize}, \vname{id}, \cdot)$ from $\party$:\\
\> Send $(\vname{msg}, \party)$ to $\mathcal{A}$ \\
\> Assert $(\vname{id}, P, P', \vname{compute}) \in \vname{Contracts}$ \\
\> If $\aname{B.Finalize}(\vname{msg}) = 1$: \\
\>\> Replace $(\vname{id}, P, P', \vname{compute})$ in $\vname{Contracts}$ \\
\>\>\> with $(\vname{id}, P, P', \vname{finalized})$ \\
\> Send blockchain state to $\mathcal{A}$
\end{tabbing}
\end{minipage} 
\end{tabular}
\end{center}
\caption{{\small
Blockchain program wrapper $B'(B)$ for PSC evaluation.\\ Note that $\mathcal{A}$ is the adversary.
}}
\label{fig:blockchain_wrapper}
\end{small}
%\end{framed}
\end{figure}

\begin{figure}[!ht]
%\begin{framed}
\begin{small}
\begin{center}
\hspace{-19pt}
User program wrapper $U'(U):$
\begin{tabular}{l}
\begin{minipage}{3in}
\begin{tabbing}
123\=12\=12\=12\=12\=\kill
\textbf{Init:} \\
\> $\vname{Computations} \gets \{\}$\\
\> $\vname{Coins} \gets \{\}$\\
\> $\vname{FrozenCoins} \gets \{\}$\\
\> Call $\aname{U.Init()}$ (i.e. Call program $U$)\\
\end{tabbing}
\end{minipage} \vspace{6pt} \\
\begin{minipage}{3in}\vspace{6pt}
\begin{tabbing}
123\=12\=12\=12\=12\=\kill
\textbf{Freeze:} Upon receiving message $(\vname{freeze}, f, P, \val, \inp)$:\\
\> $r \samplefrom \{0, 1\}^{\ell_\Com}$ \\
\> $\coin \gets \Com(\val; r)$ \\
\> $\vname{Coins} \gets \vname{Coins} \cup \{\coin\}$ \\
\> $\vname{id} \gets \aname{U.Create}(f, P)$ \\
\> $\vname{FrozenCoins} \gets \vname{FrozenCoins} \cup \{(\vname{id}, \coin)\}$ \\
\> $\vname{msg} \gets \aname{U.Freeze}(\vname{id}, (\val, r), \inp)$ \\
\> Send $\vname{msg}$ to blockchain \\
\end{tabbing}
\end{minipage} \vspace{6pt} \\
\begin{minipage}{3in}\vspace{6pt}
\begin{tabbing}
123\=12\=12\=12\=12\=\kill
\textbf{Compute:} Upon receiving message $(\vname{compute}, \vname{id})$:\\
\> $(\hat{f}, P, x) \gets \aname{U.PrepareCompute}(\vname{id})$\\
\> $\vname{Computations} \gets \vname{Computations} \cup (\hat{f}, P, \vname{id})$ \\
\> Send $(\vname{input}, \hat{f}, P, x)$ to $\mathcal{M}$
\end{tabbing}
\end{minipage} \vspace{6pt} \\
\begin{minipage}{3in}\vspace{6pt}
\begin{tabbing}
123\=12\=12\=12\=12\=\kill
\textbf{Finalize:} Upon receiving $\vname{msg} := (\vname{output}, \hat{f}, P, y)$ from $\manager$:\\
\> Assert $(\hat{f}, P, \vname{id}) \in \vname{Computations}$ \\
\> Assert $(\vname{id}, \vname{coin}) \in \vname{FrozenCoins}$ \\
\> $(\vname{msg}, \coin'_1, \hdots, \coin'_n, \vname{out}, (\val', r')) \gets \aname{U.Finalize}(\vname{id}, y)$ \\
\> $\vname{Coins} \gets (\vname{Coins} \setminus \{\coin\}) \cup \{(\val', r')\}$\\
\> Send $\vname{msg}$ to blockchain \\
\> Remove $(\hat{f}, P, \vname{id})$ from $\vname{Computations}$ \\
\> Remove $(\vname{id}, \coin)$ from $\vname{FrozenCoins}$ 
\end{tabbing}
\end{minipage} 
\end{tabular}
\end{center}
\caption{{\small
User program wrapper $U'(U)$ for PSC evaluation.
}}
\label{fig:user_wrapper}
\end{small}
%\end{framed}
\end{figure}

The fundamental correctness condition expected of a PSC evaluation protocol $\pi := (\Com, B, U)$ is as follows. Let $f$ be an $n$-ary smart contract function as defined in Definition~\ref{def:scf}. For any $\val_1, \hdots, \val_n \in \valset$ and $\inp_1, \hdots, \inp_n \in \{0, 1\}^\ast$ such that $((\val'_1, \hdots, \val'_n), \vname{out}) \gets f((\val_1, \inp_1), \hdots, (\val_n, \inp_n))$ for $\val'_1, \hdots, \val'_n \in \valset$ and $\vname{out} \in \{0, 1\}^\ast$. Then we say that $\pi$ correctly executes $\pi$ if the following experiment terminates and outputs 1 with all but negligible probability.
\begin{tabular}{l}
\begin{minipage}{3in}
\begin{tabbing}
123\=12\=12\=12\=12\=\kill
\textbf{Experiment } $\vname{Exper}_{\pi}(f, P, (\val_1, \inp_1), \hdots, (\val_n, \inp_n))$: \\
\> Run blockchain process $\mathcal{B}$ with program $B'(B)$ \\
\> Run $\mathcal{P}_i$ with program $U'(U)$ for $i \in [n]$\\
\> Send $(\vname{freeze}, f, P, \val_i, \inp_i)$ to $\mathcal{P}_i$ for $i \in [n]$ \\
\> Wait until $(\vname{id}, P, P, \vname{compute}) \in \mathcal{B}.\vname{Contracts}$ \\
\> Send $(\vname{compute}, \vname{id})$ to $\mathcal{P}_i$ for $i \in [n]$ \\
\> Wait until $(\vname{id}, P, P, \vname{finalized}) \in \mathcal{B}.\vname{Contracts}$ \\
\> Return $\sum \val'_i = \sum \val_i$ where $(\val'_i, \cdot) \in \mathcal{P}_i.\vname{Coins}$
\end{tabbing}
\end{minipage} \vspace{6pt}
\end{tabular}

To clarify on the experiment above, when a process $\mathcal{X}$ is run with a program $X$, the initialization procedure $\aname{X.Init}$ is executed first and then the control is transferred to a message loop where the program waits for incoming messages which are dispatched to their appropriate handlers. 
We denote by $\vname{transcript}_{\pi, \mathcal{S}, \mathcal{R}}$ the sequence of messages sent from $\mathcal{S}$ to $\mathcal{R}$ during the execution of protocol $\pi$. Of particular interest is $\vname{transcript}_{\pi, \mathcal{B}, \mathcal{A}}$ where $\mathcal{B}$ is the blockchain and $\mathcal{A}$ is the adversary. In fact, in our security definition, which is presented next, the objective is to simulate the protocol and produce a transcript that is computationally indistinguishable from that produced in the real world.

\subsection{Security}
We now present a security definition that is strictly weaker than security in the UC framework, which we leave to an extended paper. The security definition we describe in this section is modelled on the typical simulation-based security notion for MPC. However, there are some notable differences. Firstly, the adversary $\mathcal{A}$ interacts with a blockchain, which is modelled as an entity that is trusted for availability and correctness which shares its internal state with the adversary. Secondly, in the real world, the contract participants $\party_i$ interact privately with a trusted third party $\manager$ that cannot be corrupted by the adversary nor can the adversary access the contents of the private channels between $\manager$ and $\party_i$.

Let $I = \{i_1, \hdots, i_t\}$  be the set of indices of the $t \leq n$ corrupted parties. Additionally, we let $\bar{I} = [n] \setminus I$ be the set of indices of the honest parties. Let $f : (\valset \times \{0, 1\}^\ast)^n \to (\valset^n \times \{0, 1\}^\ast) \cup \{\bot\}$ be an $n$-party smart contract function as defined in Definition~\ref{def:scf}. We denote by $\vect{x}$ a vector of \emph{inputs} from all $n$ parties to the protocol; that is, we have $\vect{x} := (x_1 := (\val_1, \inp_1), \hdots, x_n := (\val_n, \inp_n))$. Let $\pi$ be an interactive PSC evaluation protocol as defined in Definition~\ref{def:pscprot}.
\subsubsection{Real World} We now give the experiment for the real world execution of $\pi$ in the presence of an adversary $\mathcal{A}$.

\begin{tabular}{l}
\begin{minipage}{3in}
\begin{tabbing}
123\=12\=12\=12\=12\=\kill
\textbf{Experiment } $\vname{REAL}_{\pi, \mathcal{A}, I}(f, P, \vect{x})_{\mathcal{D}}$: \\
\> Run blockchain process $\mathcal{B}$ with program $B'(B)$ \\
\> Run $\mathcal{P}_i$ with program $U'(U)$ for $i \in \bar{I}$\\
\> Parse $(\val_i, \inp_i) \gets x_i$ for $i \in \bar{I}$ \\
\> Send $(\vname{freeze}, f, P, \val_i, \inp_i)$ to $\mathcal{P}_i$ for $i \in \bar{I}$ \\
\> Wait until $(\vname{id}, P, \cdot, \cdot) \in \mathcal{B}.\vname{Contracts}$ \\
\> Send $(\vname{compute}, \vname{id})$ to $\mathcal{P}_i$ for $i \in \bar{I}$ \\
\> Output $\mathcal{D}(\vname{transcript}_{\pi, \mathcal{B}, \mathcal{A}})$ \\
\end{tabbing}
\end{minipage}
\end{tabular}

\begin{figure}[!ht]
\begin{framed}
\begin{small}
\begin{center}

\hspace{-19pt}
Idealized Entity $\mathcal{M}:$
\begin{tabular}{l}
\begin{minipage}{3in}
\begin{tabbing}
123\=12\=12\=12\=12\=\kill
\textbf{Init:} \\
\> $\vname{Inputs} \gets \{\}$
\end{tabbing}
\end{minipage} \vspace{6pt} \\
\begin{minipage}{3in}\vspace{6pt}
\begin{tabbing}
123\=12\=12\=12\=12\=\kill
\textbf{Input:} Upon receiving $(\vname{input}, F, P, x)$ from $\party$:\\
\> Assert $\party \in P$\\
\> Remove $(F, P, \party, \cdot)$ from $\vname{Inputs}$ if it exists\\
\> $\vname{Inputs} \gets \vname{Inputs} \cup (F, P, \party, x)$\\
\> If $(F, P, \party_j, x_j) \in \vname{Inputs}$ for all $\party_j \in P$:\\
\>\> $y \gets F(x_1, \hdots, x_{n'})$ where $n' = |P|$\\
\>\> Send $(\vname{output}, F, P, y)$ to $\party_j$ for all $\party_j \in P$
\end{tabbing}
\end{minipage}
\end{tabular}
\end{center}

\end{small}
\caption{{\small
Idealized entity $\mathcal{M}$ that privately and correctly executes a designated function on inputs supplied by a set of parties, which are kept private.
}}
\label{fig:m}
\end{framed}
\end{figure}

Furthermore, the blockchain executes blockchain program $B'(B)$.
Recall from the definition of $B'$ that the blockchain sends its internal state to $\mathcal{A}$ after it handles every message.

\subsubsection{Ideal World}
In the ideal world, a simulator $\mathcal{S}$ interacts with an ideal functionality $\mathcal{F}$ that executes the smart contract function $f$. All parties send their inputs to $\mathcal{F}$, then the function $f$ is computed and the result returned to $\mathcal{S}$. The goal of simulation in the ideal world is for $\mathcal{S}$ to simulate the adversary's view in the real world i.e. it must produce a transcript that is computationally indistinguishable from the transcript generated in the real world. We denote by $\vect{x}_J$ for $J \subseteq [n]$ the sub-vector of $\vect{x}$ containing the components of $\vect{x}$ at every index in $J$. 

\begin{tabular}{l}
\begin{minipage}{3in}
\begin{tabbing}
123\=12\=12\=12\=12\=\kill
\textbf{Experiment } $\aname{IDEAL}_{\mathcal{F}, \mathcal{S}, I}(f, P, \vect{x}))_{\mathcal{D}}$: \\
\> $\vname{transcript} \gets \mathcal{S}^{\mathcal{F}_{f, P}( \vect{x}_{\bar{I}}, \cdot)}(f, P, \vect{x}_{I})$\\
\> Output $\mathcal{D}(\vname{transcript})$ \\
\end{tabbing}
\end{minipage}
\end{tabular}

\begin{definition}\label{def:sec}
A PSC evaluation protocol $\pi$ is $t$-secure if for all $I \subseteq [n]$ with $|I| \leq t$, all polynomial-time computable smart contract functions $f$, all polynomial-sized $\vect{x} \in (\{0, 1\}^\ast)^n$ and all subsets of participants $P \subseteq \{\mathcal{P}_i\}_{i \in [n]}$, then for all PPT adversaries $\mathcal{A}$, there exists a PPT simulator $\mathcal{S}$ such that for any PPT distinguisher $\mathcal{D}$ it holds that
\begin{align*}
\Big|\Pr{\vname{REAL}_{\pi, \mathcal{A}, I}(f, P, \vect{x})_\mathcal{D} \rightarrow 1} -& \\ \Pr{\aname{IDEAL}_{\mathcal{F}, \mathcal{S}, I}(f, P, \vect{x})_\mathcal{D} \rightarrow 1}\Big| & \leq \fname{negl}(\lambda)
\end{align*}
where $\fname{negl}(\lambda)$ is a negligible function in the security parameter $\lambda$.
\end{definition}

\section{Our PSC Evaluation Protocol}\label{sec:PSCEP}
The protocol we present ensures both \textit{contract closure} and \textit{immediate closure}\footnote{A protocol is said to support immediate closure if the blockchain can instigate contract closure after receiving a single accepted finalization message.}. Our first idea is to instantiate the commitments used for coins with Pedersen commitments\cite{pedersen1991non}. Let $\mathbb{G}$ be a finite cyclic group of prime order $p$. Let $g$ and $h$ be two generators of $\mathbb{G}$. Then a Pedersen commitment is defined as
\[\Com(x;r) = g^x \cdot h^r\]
Consider input coins $\{\vname{coin}_j = \Com(\val_j; r_j)\}_{j \in [n]}$ and output coins $\{\vname{coin}'_j = \Com(\val'_j; s_j)\}_{j \in [n]}$. Then by the homomorphic property of Pedersen commitments\cite{groth2009homomorphic,pedersen1991non}, we have
\[\prod_{j \in [n]} \vname{coin}'_j / \vname{coin}_j = h^{\sum_{j \in [n]} s_j - r_j}\]
if and only if $\sum_{j \in [n]} \val'_j - \val_j = 0$. Our protocol leverages the Schnorr sigma protocol \cite{schnorr1989efficient} to prove knowledge of the discrete logarithm in base $h$ of $\prod_{j \in [n]} \vname{coin}'_j / \vname{coin}_j$. Note that this sigma protocol can be transformed into a non-interactive proof of knowledge in the random oracle model via the Fiat-Shamir heuristic \cite{fiat1986prove}.

\subsection{Initialization}
We define the algorithms $\aname{B.Init}$ and $\aname{U.Init}$ before we proceed any further.
 Therefore, we have
\begin{itemize}
    \item $\aname{B.Init}()$:
    \begin{itemize}
        \item
        $\vname{FreezeRecords} \gets \emptyset$
    \end{itemize}

    \item $\aname{U.Init}()$:
    \begin{itemize}
        \item $\vname{Identifiers} \gets \emptyset$
        \item
        $\vname{SecretElems} \gets \emptyset$
    \end{itemize}
\end{itemize}
% We might revisit the algorithms above later to add anything that we need to initialize; I'm leaving them almost empty for the moment

A set of $m$ participants decide to create a private smart contract by specifying an $m$-ary smart contract function $f$ along with a set of participant identities $P$ (i.e. their pseudonyms in the underlying cryptocurrency). For simplicity, we often assume that $P$ contains all participants we have defined in our security definition i.e. $m = n$ and $P = \{\mathcal{P}_i\}_{i \in [n]}$ but this is without loss of generality since it is easy to see that any subset of $m \leq n$ participants can be supported. We assume that each party gives their consent to the smart contract function $f$ that is used; that is, each party has knowledge of the code of $f$, inspects it and agrees that it is fair. Specifying and validating the correctness of $f$ is beyond the scope of this paper and literature abounds on this. For example, we would assume that each party contributes its own preconditions and post-conditions to the smart contract code in order to establish consensus. Although we do not aim to keep $f$ itself private in this paper, the astute reader will observe that it is straightforward to achieve this with our protocol. We now specify the $\aname{U.Create}$ algorithm which takes as input a smart contract function $f$ and a set of participants $P$ and returns a unique identifier for this pair. Let $H$ be a collision-resistant hash function, modelled as a random oracle in the security analysis (as we shall see later), that maps an arbitrary-length string to a uniformly random string $\{0, 1\}^{\lambda}$ where $\lambda$ is the security parameter.
\begin{itemize}
    \item $\aname{U.Create}(f, P)$:
    \begin{itemize}
        \item $\vname{id} \gets H(f \concat P)$
        \item $\vname{Identifiers} \gets \vname{Identifiers} \cup \{(\vname{id}, f, P)\}$
        \item Return $\vname{id}$.
    \end{itemize}
\end{itemize}
Much of what we have defined thus far is ``boilerplate" and possibly common to any PSC evaluation protocol. Next, we define the algorithms used in the freeze phase.

\subsection{Freezing Coins}
In the freeze phase of the protocol, we employ the sigma protocol from \cite{groth2015one} that proves in zero knowledge that a commitment commits to $0$ or $1$. We use its transformation (in the random oracle model) into a NIZK via the Fiat-Shamir heuristic \cite{fiat1986prove}, and denote the scheme by $\vname{BNIZK} := (\aname{BNIZK.Prove}, \aname{BNIZK.Verify})$. We have already established that $\valset$ denotes the set of valid amounts of currency. Here we ordain the upper limit $L$ of a valid amount of currency to be $2^\ell$ such that a valid amount of currency lies in the range $[0, 2^{\ell})$ for some $\ell$. In addition, we decompose an output coin $\coin'$ into a product of $\ell$ commitments corresponding to the binary decomposition of the output value. For each $k \in \{0, \hdots, \ell - 1\}$, the party generates two commitments, one of them commits to zero and the other commits to one. Importantly, it permutes the order of these commitments when it sends them to the blockchain so that the blockchain does not know whether the first or the second one commits to zero etc. More precisely, for each $k \in \{0, \hdots, \ell - 1\}$, we sample a bit $b_k \samplefrom \{0, 1\}$ and compute $c^{(k, b_k)}$ as a commitment to $b_k$ and $c^{(k, 1 - b_k})$ as a commitment to $1 - b_k$. We prove in zero knowledge that both commitments commit to a bit using $\vname{BNIZK}$ and sends both commitments along with the proofs to the blockchain. These steps take place in the $\aname{U.Freeze}$ algorithm, which we now formally define:
\begin{itemize}
    \item $\aname{U.Freeze}(\vname{id}, (\val, r), \inp):$
\begin{itemize}
\item
Assert $(\vname{id}, f, P) \in \vname{Identifiers}$
\item
$\coin \gets \Com(\val; r)$
    \item For $k \in \{0, \hdots, \ell - 1\}$:
    \begin{itemize}
        \item
        $b_k \samplefrom \{0, 1\}$
        \item
        $s^{(k, 0)} \samplefrom \ZZ_p$
        \item
        $s^{(k, 1)} \samplefrom \ZZ_p$
        \item
        $c^{(k, b_k)} \gets \Com(b_k; s^{(k, b_k)})$
        \item
        $c^{(k, 1 - b_k)} \gets \Com(1 - b_k; s^{(k, 1 - b_k)})$
        \item
        $\pi^{(k, b_k)} \gets \aname{BNIZK.Prove}(c^{(k, b_k)}, (s^{(k, b_k)}, b_k))$
        \item
        $\pi^{(k, 1 - b_k)} \gets \aname{BNIZK.Prove}(c^{(k, 1 - b_k)}, (s^{(k, 1 - b_k)}, \\ 1 - b_k))$
    \end{itemize}
    \item
    $\vname{SecretElems} \gets \vname{SecretElems} \cup \{(\vname{id}, \val, r, \inp,\\ \{(c^{(k, 0)}, c^{(k, 1)}, s^{(k, 0)}, s^{(k, 1)})\}_{0 \leq k < \ell})\}$
    \item
        $
        \begin{array}{ll}
        \vname{msg} \gets
        (\vname{freeze}, \vname{id}, P, \vname{coin}, \\
        \{((c^{(k, b_k)}, \pi^{(k, b_k)}), (c^{(k, 1 - b_k)}, \pi^{(k, 1 - b_k)}))\}_{0 \leq k < \ell})
        \end{array}$
        \item
        Return $\vname{msg}$
\end{itemize}
\end{itemize}

where $\vname{id}$ is the contract identifier, $P$ is the set of contract participants i.e. $\{\party_i\}_{i \in [n]}$.

The blockchain responds by executing the following steps, which are part of the $\aname{B.Freeze}$ algorithm that we now define:
\begin{itemize}
    \item $\aname{B.Freeze}(\vname{msg}, \party)$:
\begin{itemize}
    \item
    Parse $\vname{msg}$ as
    
    $\begin{array}{cc}
        (\vname{freeze}, \vname{id}, P, \vname{coin}, \\
         \{((c^{(k, b_k)}, \pi^{(k, b_k)}), (c^{(k, 1 - b_k)}, \pi^{(k, 1 - b_k)}))\}_{0 \leq k < \ell})
        \end{array}$
    \item
    For $k \in \{0, \hdots, \ell - 1\}$:
    \begin{enumerate}
        \item Assert $\aname{BNIZK.Verify}(c^{(k, b_k)}, \pi^{(k, b_k)})$
        \item
        Assert $\aname{BNIZK.Verify}(c^{(k, 1 - b_k)}, \pi^{(k, 1 - b_k)})$
        \item
        $C^{(k)} \gets \{c^{(k, b_k)}, c^{(k, 1 - b_k)}\}$
    \end{enumerate}
    \item
    $\vname{FreezeRecords} \gets \vname{FreezeRecords} \cup \{(\vname{id}, \party, \coin,\\ \{C^{(k)}\}_{0 \leq k < \ell})\}$
    \item
    Return $1$
\end{itemize}
\end{itemize}

\subsection{Computation}
Consider a smart contract function $f: (\valset \times\{0,1\}^{*})^{n} \rightarrow \valset^{n}\times\{0,1\}^{*}$. Let $\val_i$ and $\inp_i$ be the input currency value and input string of party $\party_i$ respectively.\\
All parties obtain the tuples $(\vname{id}, \party_i, \coin_i, \{C_i^{(k)}\}_{0 \leq k < \ell})$ from the blockchain.

We now define the following function $\hat{f}$. We use the notation $\val[k]$ to denote the $k$-th bit of $\val$.

\begin{figure}[!ht]
\begin{framed}
\begin{small}
\begin{center}
\textbf{MPC Function $\hat{f}$}
\hspace{-19pt}
\begin{tabular}{l}
\begin{minipage}{3in}
\begin{tabbing}
123\=12\=12\=12\=12\=\kill
 \\
$\hat{f}((\val_1, r_1, \inp_1, (c_1^{(k, 0)}, c_1^{(k, 1)}, s_1^{(k, 0)}, s_1^{(k, 1)})_{0 \leq k < \ell})$, \\
\>    $\hdots, (\val_n, r_n, \inp_n, (c_n^{(k, 0)}, c_n^{(k, 1)}, s_n^{(k, 0)}, s_n^{(k, 1)})_{0 \leq k < \ell} ))$\\
\> $((\val'_1, \hdots, \val'_n), \vname{out}) \gets $\\
\>\>\>        $f((\val_1, \inp_1), \hdots, (\val_n, \inp_n))$ \\
\> For all $j \in [n]$: \\
\>\> For all $k \in \{0, \hdots, \ell - 1\}$:\\
\>\>\> $c_j^{(k)} \gets c_j^{(k, \val'_j[k])}$ \\
\> $r \gets \sum_{j \in [n]} (\sum_{0 \leq k < \ell} 2^k \cdot s_j^{(k, \val'_j[k])}- r_j)$ \\
\> Compute proof $\pi$ with Schnorr protocol with\\
\>\>\>knowledge $r$ and base $h$ \\
\> Return $(((c_1^{(k)})_{0 \leq k < \ell}, \hdots, (c_n^{(k)})_{0 \leq k < \ell}), \pi, \vname{out})$
\end{tabbing}
\end{minipage}
\end{tabular}
\end{center}

\end{small}
\caption{{\small
Definition of the function that is evaluated by the MPC protocol. This function is built from $f$, the smart contract function.
}}
\label{fig:mpcfunc}
\end{framed}
\end{figure}

We can observe two things here. Firstly, all parties must collaborate to produce an accepting proof $\pi$, which serves to validate the zero-sum constraint. The second point is that symmetric encryption of the $\val'_i$ is not needed since the value $\val'_i$ can be read by $\party_i$ (and only $\party_i$) from $((c_i^{(k)})_{0 \leq k < \ell}$.

Now we define the $\aname{U.PrepareCompute}$ algorithm:
\begin{itemize}
    \item $\aname{U.PrepareCompute}(\vname{id})$:
    \begin{itemize}
    \item
    Assert $(\vname{id}, f, P) \in \vname{Identifiers}$
    \item
    Assert $(\vname{id}, \val, r, \inp, \{(c^{(k, 0)}, c^{(k, 1)}, s^{(k, 0)}, s^{(k, 1)})\}_{0 \leq k < \ell}) \\
    \in \vname{SecretElems}$
        \item Fetch $(\vname{id}, \party_i, \coin_i, \{C_i^{(k)}\}_{0 \leq k < \ell})$ from $\mathcal{B}.\vname{FreezeRecords}$ for each $\party_i \in P$ (recall we assume that $|P| = n$)
        \item
        Define function $\hat{f}$ as in Figure~\ref{fig:mpcfunc} using function $f$.
        \item
        $x \gets (\val, r, \inp, (c^{(k, 0)}, c^{(k, 1)}, s^{(k, 0)}, s^{(k, 1)})_{0 \leq k < \ell})$
        \item
        Return $(\hat{f}, P, x)$
    \end{itemize}
\end{itemize}

\subsection{Finalization}
Since this protocol supports immediate closure, a valid finalization message from a single party is sufficient to instigate contract closure. The MPC protocol is executed between the parties during the computation phase and it concludes with a final output $y$ that is accessible to all the parties. Therefore, we now define the algorithm, namely $\aname{U.Finalize},$ that processes this $y$ and produces a $\vname{finalize}$ message to be sent to the blockchain along with extracting all output coins $\coin'_1, \hdots, \coin'_n$.
\begin{itemize}
    \item $\aname{U.Finalize}(\vname{id}, y)$:
    \begin{itemize}
    \item
    Assert $(\vname{id}, f, P) \in \vname{Identifiers}$
    \item
    Assert $(\vname{id}, \val, r, \inp,\\  \{(c^{(k, 0)}, c^{(k, 1)}, s^{(k, 0)}, s^{(k, 1)})\}_{0 \leq k < \ell})
    \in \vname{SecretElems}$
    \item Fetch $(\vname{id}, \party_i, \coin_i, \{C_i^{(k)}\}_{0 \leq k < \ell})$ from $\mathcal{B}.\vname{FreezeRecords}$ for each $\party_i \in P$ (recall we assume that $|P| = n$)
    \item
    Parse $y$ as $(((c_1^{(k)})_{0 \leq k < \ell}, \hdots, (c_n^{(k)})_{0 \leq k < \ell}), \pi, \vname{out})$
    For all $j \in [n]$:
    \begin{enumerate}
        \item 
        For all $k \in \{0, \hdots, \ell - 1\}$:
        \begin{enumerate}
            \item Assert $c_j^{(k)} \in C_j^{(k)}$
        \end{enumerate}
        \item
        $\vname{coin}'_j \gets \prod_{0 \leq k < \ell} (c_j^{(k)})^{2^k}$ 
    \end{enumerate}
    \item
    $c \gets \prod_{j \in [n]} \vname{coin}'_j / \vname{coin}_j$
    \item
    Assert $\aname{Schnorr.Verify}((c, h), \pi)$
    \item
    Let $i'$ be the index of \emph{this} party
    \item
    For $k \in \{0, \hdots, \ell - 1\}$:
    \begin{enumerate}
        \item If $c_{i'}^{(k)} = c_{i'}^{(k, 0)}$:
        \begin{itemize}
            \item $v_k \gets 0$
        \end{itemize}
        Else:
        \begin{itemize}
            \item $v_k \gets 1$
        \end{itemize}
    \end{enumerate}
    \item
    $\val' \gets \sum_{0 \leq k < \ell}v_k \cdot 2^k$
    \item
    $r' \gets \sum_{0 \leq k < \ell} s_{i'}^{(k, v_k)} \cdot 2^k$
    \item
    $\coin'_{i'} \gets \Com(\val',r')$
    \item
    $\vname{msg} \gets (\vname{finalize}, \vname{id}, ((c_1^{(k)})_{0 \leq k < \ell}, \hdots, (c_n^{(k)})_{0 \leq k < \ell}),\\ \vname{out}, \pi)$
    \item
    Return $(\vname{msg}, \coin'_1, \hdots, \coin'_n, \vname{out}, (\val', r'))$
    \end{itemize}
\end{itemize}

The message that is returned in the algorithm above is sent by the contract participant to the blockchain, which processes it in the following algorithm $\aname{B.Finalize}$:
\begin{itemize}
    \item $\aname{B.Finalize}(\vname{msg})$:
    \begin{itemize}
    \item
     Parse $\vname{msg}$ as $(\vname{finalize}, \vname{id}, \\ ((c_1^{(k)})_{0 \leq k < \ell}, \hdots, (c_n^{(k)})_{0 \leq k < \ell}), \vname{out}, \pi)$
     %\item
     %Assert $(\vname{id}, P, P', \vname{compute}) \in \vname{Contracts}$
    \item Assert $(\vname{id}, \party_i, \coin_i, \{C_i^{(k)}\}_{0 \leq k < \ell}) \in \vname{FreezeRecords}$ for each $\party_i \in P$ (recall we assume that $|P| = n$)
    \item
    For all $j \in [n]$:
    \begin{enumerate}
        \item 
        For all $k \in \{0, \hdots, \ell - 1\}$:
        \begin{enumerate}
            \item Assert $c_j^{(k)} \in C_j^{(k)}$
        \end{enumerate}
        \item
        $\vname{coin}'_j \gets \prod_{0 \leq k < \ell} (c_j^{(k)})^{2^k}$ 
    \end{enumerate}
    \item
    $c \gets \prod_{j \in [n]} \vname{coin}'_j / \vname{coin}_j$
    \item
    Return $\aname{Schnorr.Verify}((c, h), \pi)$
    \end{itemize}
\end{itemize}

\section{Security Analysis}\label{sec:security}
In this section we give a proof of security for our PSC evaluation protocol $\pi := (\Com, B, U)$ defined in the previous section. More precisely, we prove that $\pi$ is $t$-secure for all $t < n$ in accordance with Definition~\ref{def:sec} assuming the hardness of the discrete logarithm problem in the random oracle model. We must construct a simulator $\mathcal{S}$ that uses an adversary $\mathcal{A}$ to simulate a transcript that is computationally indistinguishable to any PPT distinguisher $\mathcal{D}$ from the transcript produced in the real world execution of $\pi$. To correctly make use of $\mathcal{A}$ in the simulation, we must correctly simulate $\mathcal{A}$'s view which requires us to simulate the honest parties, the blockchain and $\mathcal{M}$. The simulator is permitted access to an ideal functionality $\mathcal{F}$ that computes the smart contract function. Note that the current version of our protocol, for simplicity, does not provide verification for the smart contract output string $\vname{out}$ nor does our correctness condition check its veracity. We note that it is straightforward to modify the protocol such that the output string is included in the Schnorr proof such that the resulting proof serves to encompass a ``vote" amongst all participants of the correctness of $\vname{out}$ since the MPC program is assumed to be correct.

\begin{theorem}
Assuming a $t$-secure MPC protocol, our PSC evaluation protocol is $t$-secure for all $t < n$ assuming the hardness of the discrete logarithm problem in the random oracle model.
\end{theorem}
\begin{proof}
We prove the theorem via a hybrid argument. Our goal is to move to a hybrid where we no longer depend on the inputs of the honest parties i.e. $(x_i := (\val_i, \inp_i))_{i \in \bar{I}}$. Recall that $I$ denotes the set of indices of the corrupt parties whereas $\bar{I}$ denotes the set of indices of the honest parties. The reductions underlying indistinguishability of the various hybrids are relatively straightforward and for brevity sake, we omit them here.

\noindent \textbf{Hybrid 0}: This is the real system.

\vspace{1pt}

\noindent \textbf{Hybrid 1}: The change we make in this hybrid is to simulate $\mathcal{M}$ and its output $y$ that it sends to all parties. We receive the $(\vname{input}, \cdot)$ messages from the adversary for the corrupted parties and obtain the inputs of the corrupted parties. In this hybrid, we still have access to the inputs of the honest parties. We obtain $y$ be computing the function $\hat{f}$ in Figure~\ref{fig:mpcfunc} and send it to all parties as an ($\vname{output,\cdot}$) message, therefore simulating $\mathcal{M}$.
This hybrid is distributed identically to the previous hybrid.

\vspace{1pt}

\noindent \textbf{Hybrid 2}: The change we make in this hybrid is to the Schnorr proof $\pi$ component of $y$ (i.e. the output of $\hat{f}$). Instead of computing the proof normally, we use the simulator for the NIZK. Therefore, we no longer need the randomness $r_i$ and $s_i^{(k, b)}$ to generate $\pi$ and it can still be generated even when the zero-sum condition does not hold.
The hybrids are indistinguishable from the zero-knowledge property of the Schnorr protocol.

\vspace{1pt}

\noindent For $i \in \bar{I}$ and $k \in \{0, \hdots, \ell - 1\}$:

\noindent \textbf{Hybrid 3, i, k}: In this hybrid, we do not compute the proofs $\pi_i^{(k, 0)}$ and $\pi_i^{(k, 1)}$ using $\aname{BNIZK.Prove}$ but instead use the simulator for $\vname{BNIZK}$.
Indistinguishability of the hybrids follows form the zero-knowledge property of $\vname{BNIZK}$.

\vspace{1pt}

\noindent For $i \in \bar{I}$ and $k \in \{0, \hdots, \ell - 1\}$ and $b \in \{0, 1\}$:

\noindent \textbf{Hybrid 4, i, k, b}: In this hybrid, we compute the commitment $c_i^{(k, b)}$ as a commitment to a uniformly random element of $\ZZ_p$.
Indistinguishability of the hybrids from the perfectly blinding property of Pedersen commitments.

\vspace{1pt}

\noindent For $i \in \bar{I}$:

\noindent \textbf{Hybrid 5, i}:
In this hybrid, we change the input coin commitment $\coin_i$ to a commitment to a uniformly random element of $\ZZ_p$.
Indistinguishability of the hybrids from the perfectly blinding property of Pedersen commitments.

\vspace{1pt}

\noindent \textbf{Hybrid 6}: In this hybrid, we change how we compute $(\val'_i)_{i \in I}, \vname{out})$ by instead calling the ideal functionality $\mathcal{F}$ and sending it the inputs of the corrupted parties, which have been obtained from $\mathcal{A}$.
This hybrid is distributed identically to the previous hybrid.
This hybrid has no dependence on the inputs and outputs of the honest parties and hence we obtain a simulator that uses the adversary $\mathcal{A}$ to produce a transcript that is computationally indistinguishable from the real world (Hybrid 0).  We have shown that privacy holds. In regard to correctness, if $\mathcal{A}$ sends a finalization message to the blockchain with a proof that is accepted where the zero-sum constraint does not hold, then we can construct an algorithm that uses $\mathcal{A}$ to produce a proof that violates soundness of the Schnorr NIZK, which contradicts the statement of the theorem. The result follows.
\end{proof}

\section{Conclusion and Future Work with NI-PSC}\label{sec:nipsc}
The protocol presented above facilitates Interactive Private Smart Contracts (I-PSC) wherein an Interactive protocol (MPC) between the parties is executed replacing the manager in Hawk in order to determine the outcome of the smart contract and reach consensus on the payout distribution. At least one of the parties must then notify the blockchain of the outcome so that the blockchain can instigate contract closure. A disadvantage of I-PSC is that all parties must remain online for computation to proceed. A more desirous and powerful alternative is a protocol facilitating Non-Interactive Private Smart Contracts (NI-PSC) wherein the blockchain performs the computation of the smart contract while retaining input privacy; that is, its computation is oblivious to the inputs. Parties need not be online during computation nor are they required to interact. Furthermore, there is no possibility for aborts as in the interactive case and the execution of smart contracts can be scheduled for specified times and use public information such as stock prices, a situation that may lead to aborts in the interactive case e.g. a stock price that negatively affects a party's position in a contract may encourage that party to abort the protocol. In a NI-PSC, only one round of interaction with the blockchain is needed to execute a private smart contract, namely the freeze round. In the freeze round, each party sends a message containing (1). a unique identifier $\vname{id}$ that identifies the contract (e.g. a hash of the contract function, execution number and set of participants); (2). the set of participants $P$; (3). the party's input coin; and (4). protocol-specific information. The round is complete when the blockchain receives a valid freeze message with matching $(\vname{id}, P)$ for every $\party \in P$. The blockchain can then independently compute the output of the smart contract and instigate contract closure.

\subsection{Realizing NI-PSC}
NI-PSC can be realized with a primitive known as Non-Interactive MPC (NI-MPC). We assume we are in the PKI (public key infrastructure) model where the parties know each other's authentic public key. Then, in NI-MPC, each party need only send one message, which contains its input and the function to be evaluated. Given the messages from all parties, the desired function can be computed on the parties' inputs. Choosing the function $f'$ described in our protocol above as the function computed by NI-MPC would fully realize NI-PSC.  Known constructions of NI-MPC rely on indistinguishability obfuscation (iO) \cite{barak2001possibility,ananth2015indistinguishability}, which has not been realized from standard assumptions. Basing NI-MPC on standard assumptions is a stubbornly difficult open problem because NI-MPC for general functions implies iO. Since ``pure" NI-PSC is so difficult to achieve under standard assumptions, we intend to consider various relaxations as part of future work in order to give constructions under standard assumptions.

\bibliographystyle{./IEEEtran}
\bibliography{./zkhawk}
   % \onecolumn
    %\section{Notation}\label{NO}
    %\vspace{10mm}
    %\begin{center}
    %\begin{tabular}{| c | c |} 
    %\hline
    %$a\overset{\$}{\gets}B$ & This notation denotes that $a$ is sampled according to the distribution B \\
    %\hline
    %$[n]$ & The set of elements $\{1,...,n\}$ \\
    %\hline
    %$\vect{a}$ & Vector quantity (always in bold) \\
    %\hline
    %$A_1 \cind A_2$ & The distributions $A_1$ and $A_2$ are computationally indistinguishable \\
    %\hline
    %$C \gets \Com(x,r)$ & $C$ is a commitment for the value $x$ and randomness $r$ \\
    %\hline
    %\end{tabular}
    %\end{center}

\end{document}